\pdfoutput=1
\documentclass[11pt]{article}
\usepackage[letterpaper, margin=1in]{geometry}
\usepackage[utf8]{inputenc}
\usepackage[T1]{fontenc}
\usepackage{lmodern}
\usepackage{graphicx}
\usepackage{amsmath}
\usepackage{amssymb}
\usepackage{amsthm}
\usepackage{booktabs}
\usepackage[hidelinks]{hyperref}

\hypersetup{
	pdftitle={Barrier Functions Against Safety Drift in Shared Control},
	pdfauthor={M. Yusuf Uzun, Yildiray Yildiz}
}

\newcommand{\norm}[1]{\lVert#1\rVert}
\newcommand{\tr}{\mathsf{T}}
\newcommand{\Reals}{\mathbb{R}}
\newcommand{\Wo}{W_{\!o}}
\newcommand{\Wod}{W_{\!o}^{\dagger}}

\theoremstyle{plain}
\newtheorem{theorem}{Theorem}
\newtheorem{lemma}{Lemma}
\newtheorem{proposition}{Proposition}
\theoremstyle{definition}
\newtheorem{definition}{Definition}
\newtheorem{assumption}{Assumption}
\theoremstyle{remark}
\newtheorem*{remark}{Remark}

\title{Barrier Functions Against Safety Drift in Shared Control}
\author{
	M. Yusuf Uzun\textsuperscript{1} \qquad Yildiray Yildiz\textsuperscript{2} \\[6pt]
	\small \textsuperscript{1}Department of Mechanical, Materials, and Aerospace Engineering \\
	\small Illinois Institute of Technology, Chicago, IL 60616, USA \\
	\small \texttt{muzun@hawk.illinoistech.edu} \\[4pt]
	\small \textsuperscript{2}Department of Mechanical Engineering \\
	\small Bilkent University, 06800 Ankara, T\"{u}rkiye \\
	\small \texttt{yyildiz@bilkent.edu.tr}
}
\date{}

\begin{document}

\maketitle

\begin{abstract}
	Shared-control arbitration mechanisms can fail structurally when they are driven by quantities that evolve on the same automation-assisted trajectory that unsafe automation can corrupt. Under unsafe assistance, arising from faults or flawed design, the arbitration logic can adapt to the corrupted trajectory instead of resisting it. Focusing on workload regulation, this paper proposes a framework that prevents this problem by anchoring the safety constraints to an unassisted pilot trajectory. A baseline model evolving without automation assistance generates the acceptable workload envelope, while an assistance-induced deviation barrier and a closed-form mismatch bound relate actual and baseline workload evolution. The resulting control barrier function based quadratic program is always feasible and guarantees forward invariance of an augmented safe set. Simulations on an F-16 pitch-tracking problem show that the proposed method prevents workload-constraint drift under unsafe assistance.
\end{abstract}

\section{Introduction}
\label{sec:intro}

The human--automation interaction literature has long shown that a central challenge for a harmonious cooperation is inappropriate feedback and interaction: automatic systems can compensate for developing abnormalities while providing too little information for timely diagnosis and recovery~\cite{norman1990problem}. Woods and Sarter similarly described ``going sour'' accidents as coordination breakdowns that reveal deeper patterns in human--machine cooperation rather than isolated glitches or simple operator mistakes~\cite{woods2000going}. The Boeing 737 MAX Maneuvering Characteristics Augmentation System failures provide a vivid example of such escalation under unsafe human--automation interaction, where a static assignment of control authority, combined with faulty sensor data, leads to repeated undesired autonomous intervention with catastrophic consequences~\cite{curran2024mcas}.

One important response to these challenges is better arbitration. In shared control, arbitration regulates how authority is balanced between human and automation, often by adjusting assistance online according to operator-centred quantities such as operator state~\cite{sentouh2018driver, benloucif2019online, nguyen2018sensor}, trust~\cite{hu2021trust}, or cooperation state (e.g., agreement, conflict, or negotiation between human and automation)~\cite{wang2020decision, uzun2024robust}. These schemes are attractive because they adapt assistance to the condition of the human--automation team rather than relying on a fixed blending law. However, they typically rely on hand-tuned, parameter-dependent blending functions tailored to specific system architectures. A recent control barrier function (CBF) based work~\cite{uzun2025arbitration} has introduced an alternative: by expressing shared-control metrics such as workload as real-time inequality constraints in a quadratic program (QP), the arbitration becomes systematic (no parameter-dependent blending functions) and modular (new metrics can be added as additional barriers without changing the arbitration structure), and feasibility is guaranteed because disengaging automation always satisfies all constraints, preserving pilot authority by construction. However, in both traditional and CBF-based schemes, the regulating quantity is at least partly inferred from the ongoing assisted interaction. Under nominal conditions this is effective, but it also creates a vulnerability. In case of persistent unsafe assistance, automation affects not only the system state but also the very quantities used to decide how much assistance should be allowed. The mechanism can therefore create a vicious cycle: unsafe assistance alters the trajectory, the altered trajectory changes the arbitration variables, and the updated arbitration logic can in turn continue to admit the unsafe assistance. In this sense, the regulating constraints themselves can ``go sour''~\cite{woods2000going}.

This paper addresses the problem of unsafe assistance-induced instability in shared control systems discussed above, in the context of workload regulation. Following the arbitration control barrier function (ACBF) framework of~\cite{uzun2025arbitration}, we decouple the safety envelope from the assistance it regulates by constructing an operator-only baseline model, namely a copy of the coupled human--plant dynamics evolving under zero automation assistance, and generate the adaptive workload envelope from this baseline trajectory rather than the actual system. An assistance-induced deviation barrier confines the actual trajectory near the baseline, and a computable mismatch bound relates actual and baseline workload evolution. The resulting ACBF-based safety filter is always feasible because automation disengagement is admissible by construction, and makes the augmented safe set forward invariant. The method is evaluated on an F-16 pitch-tracking task under four severe unsafe assistance scenarios.

\section{Framework}
\label{sec:framework}

\subsection{Coupled Human--Plant Dynamics}

We consider a control-affine plant
\begin{equation} \label{eq:pdyn}
	\dot{x}_p = f_p(x_p) + g_p(x_p)\,u_p,
\end{equation}
where $x_p \in \Reals^{n_p}$, $u_p \in \Reals^{q_p}$, and $f_p$, $g_p$ are locally Lipschitz. The plant input is the sum of human and automation commands,
\begin{equation}\label{eq:up}
	u_p = y_h + y_a,
\end{equation}
where $y_h,\, y_a \in \Reals^{q_p}$ denote the human operator command and the automation assistance. The human operator is modelled as
\begin{subequations}\label{eq:hdyn}
	\begin{align}
		\dot{x}_h &= f_h(x_h,x_p) + g_h(x_h,x_p)\, r, \\
		y_h &= h_h(x_h) + j_h(x_h)\, r,
	\end{align}
\end{subequations}
where $x_h \in \Reals^{n_h}$ is the human state, $r \in \Reals^{q_h}$ is the reference signal with $q_h$ the reference dimension, and $f_h$, $g_h$, $h_h$, $j_h$ are locally Lipschitz mappings of appropriate dimensions. Using~\eqref{eq:pdyn}--\eqref{eq:hdyn}, the aggregated state $x = [x_p^\tr \;\; x_h^\tr]^\tr \in \Reals^{n}$, $n = n_p + n_h$, evolves as
\begin{equation} \label{eq:cphsdyn}
	\dot{x} = f(x) + g_a(x)\,y_a + g_r(x)\,r,
\end{equation}
where $f$, $g_a$, and $g_r$ collect the unforced, automation-input, and reference-input contributions identified from \eqref{eq:pdyn}~--~\eqref{eq:hdyn}. The reference $r$ and its derivatives up to order $\nu \in \mathbb{N}$ are assumed available, with $r \in C^\nu$. The required derivative order $\nu$ will be determined by the safety constraints introduced in Section~\ref{sec:design}.

\subsection[Arbitration Control Barrier Functions (ACBFs) and Safety-Filter Quadratic Program (QP)]{Arbitration Control Barrier Functions (ACBFs) and\\ Safety-Filter Quadratic Program (QP)}

A barrier function is a scalar function whose sign encodes whether the system is inside a desired safe region. To enforce safety, one must monitor how this function evolves in time. Let $\rho \triangleq (r, \dot{r}, \ldots, r^{(\nu)})$ denote the reference signal and its derivatives introduced in Section~\ref{sec:framework}, and let $b(x,\rho) : \Reals^n \times \Reals^{(\nu+1)q_h} \to \Reals$ be a continuously differentiable candidate barrier function. We denote its time derivative along the no-assistance dynamics ($y_a = 0$) and along the assisted dynamics, respectively, as
\begin{subequations}\label{eq:derivatives}
	\begin{align}
		\mathcal{L}_0 b(x,\rho) &:= L_f b + L_{g_r} b\, r + \textstyle\sum_{i=0}^{\nu-1} \frac{\partial b}{\partial r^{(i)}}\, r^{(i+1)}, \label{eq:L0}\\
		\mathcal{L}_a b(x,\rho;y_a) &:= \mathcal{L}_0 b(x,\rho) + L_{g_a} b(x,\rho)\, y_a, \label{eq:La}
	\end{align}
\end{subequations}
where $L_f b \triangleq \frac{\partial b}{\partial x} f(x)$, $L_{g_r} b \triangleq \frac{\partial b}{\partial x} g_r(x)$, $L_{g_a} b \triangleq \frac{\partial b}{\partial x} g_a(x)$ are Lie derivatives with respect to the unforced dynamics $f$, reference-input $g_r$, and automation-input $g_a$ from~\eqref{eq:cphsdyn}. The operator $\mathcal{L}_0 b$ captures the barrier's unforced dynamics under pilot-only dynamics, while $\mathcal{L}_a b$ adds the automation's contribution through $L_{g_a} b\, y_a$.

\begin{definition}[ACBF~\cite{uzun2025arbitration}]\label{def:acbf}
	Let $C(\rho) = \{x \in \Reals^n \mid b(x,\rho) \ge 0\}$. The function~$b$ is an \emph{arbitration control barrier function} (ACBF) for~\eqref{eq:cphsdyn} if there exists a class-$\mathcal{K}$ function $\alpha$ such that
	\begin{equation}\label{eq:acbf}
		\mathcal{L}_0 b(x,\rho) + \alpha\!\left(b(x,\rho)\right) \ge 0,
		\quad \forall\, x \in C(\rho).
	\end{equation}
	The ACBF condition is evaluated at $y_a = 0$, ensuring that automation disengagement is always feasible.
\end{definition}

Given ACBFs $b_1,\ldots,b_m$ with the safe set $C(\rho) = \bigcap_{i=1}^{m}\{x : b_i(x,\rho) \ge 0\}$, the filtered automation assistance is computed as
\begin{align}\label{eq:qp}
	y_a^\star(t)
	&= \arg\min_{y_a \in Y_a}\; \norm{y_a - y_{a,\mathrm{nom}}(t)}^2 \\
	&\quad \text{s.t.}\;\;
	\mathcal{L}_a b_i(x,\rho;y_a) + \alpha_i(b_i(x,\rho)) \ge 0,\;\; i=1,\ldots,m, \nonumber
\end{align}
where $y_{a,\mathrm{nom}}(t)$ is the unfiltered (pre-safety-filter) output of the automation assistance system and $Y_a \subseteq \Reals^{q_p}$ is the admissible assistance set containing the origin. The term $\mathcal{L}_a b_i$ in each constraint is defined in~\eqref{eq:La}, and since $L_{g_a} b_i\, y_a$ is affine in $y_a$, \eqref{eq:qp} is a convex QP. Feasibility is guaranteed because setting $y_a = 0$ (i.e., disengaging automation) always satisfies all constraints by the ACBF condition~\eqref{eq:acbf}.

\section[Pilot-Anchored Arbitration Control Barrier Function (ACBF) Design]{Pilot-Anchored Arbitration Control\\ Barrier Function (ACBF) Design}
\label{sec:design}

\subsection{Linear Shared-Control Dynamics}

We now specialize the general nonlinear framework of Section~\ref{sec:framework} to linear plants and a linear pilot model. The linear specialization below enables closed-form workload barriers and a computable mismatch bound.

Consider a linear plant with state $x_p \in \Reals^{n_p}$. The plant input combines the external command $u_p = y_h + y_a$ as in~\eqref{eq:up} with state feedback through a gain $K_p \in \Reals^{1 \times n_p}$, yielding the closed-loop dynamics
\begin{equation}\label{eq:plant_lin}
	\dot{x}_p = A_p\,x_p + B_p\,u_p, \quad y_p = C_p\,x_p,
\end{equation}
where $u_p$ is the total plant input, $A_p = A_p^{\mathrm{ol}} - B_p K_p \in \Reals^{n_p \times n_p}$ is the stable closed-loop plant matrix with $A_p^{\mathrm{ol}} \in \Reals^{n_p \times n_p}$ being the open-loop plant matrix, $B_p \in \Reals^{n_p \times 1}$ is the input matrix, and $C_p \in \Reals^{1 \times n_p}$ is the output matrix. The human operator model's transfer function relating the pilot output to the tracking error is \cite{bacon1983optimal}
\begin{equation}\label{eq:humantf}
	\frac{Y_h(s)}{R(s) - Y_p(s)} = k_p\,\frac{T_p s + 1}{T_z s + 1},
\end{equation}
where $Y_h(s)$, $R(s)$, and $Y_p(s)$ are human input, reference signal, and plant output, $k_p$ is the pilot's proportional gain, $T_p$ is the lead time constant (reflecting anticipatory response), and $T_z$ is the lag time constant (reflecting neuromuscular delay). A state-space realization of \eqref{eq:humantf} gives
\begin{subequations}\label{eq:pilot_ss}
	\begin{align}
		\dot{x}_h &= a_h\, x_h + b_h(r - y_p), \label{eq:pilot_ss_a}\\
		y_h &= c_h\, x_h + d_h(r - y_p), \label{eq:pilot_ss_b}
	\end{align}
\end{subequations}
where $x_h \in \Reals$ is the pilot state, $y_p = C_p x_p$ is the plant output from~\eqref{eq:plant_lin}, and $(a_h, b_h, c_h, d_h)$ are the scalar parameters.

Collecting the plant dynamics~\eqref{eq:plant_lin} and the pilot model~\eqref{eq:pilot_ss} into a single state vector $x = [x_p^\tr \;\; x_h]^\tr \in \Reals^{n}$, $n = n_p + 1$, gives the aggregated linear system
\begin{equation}\label{eq:lindyn}
	\dot{x} = A\,x + B_{y_a}\,y_a + B_r\,r,
\end{equation}
where $B_{y_a} \in \Reals^{n\times 1}$ and $B_r \in \Reals^{n\times 1}$ are the automation and reference input matrices, and the aggregated system matrix is
\begin{equation}\label{eq:Amat}
	A = \begin{bmatrix} A_p^{\mathrm{ol}} - B_p K_p - B_p d_h C_p & B_p c_h \\ -b_h C_p & a_h \end{bmatrix} \in \Reals^{n \times n}.
\end{equation}
Considering a well-trained pilot, $A$ in~\eqref{eq:Amat} is assumed Hurwitz. The pilot workload is defined as
\begin{equation}\label{eq:wp}
	w_p(x,\rho) \triangleq y_h^2 + \dot{y}_h^2,
\end{equation}
where $y_h$ is given in~\eqref{eq:pilot_ss_b} and $\dot{y}_h$ is its rate.

\subsection{Baseline Model and Workload Barriers}
\label{sec:predictor}

To decouple the workload envelope from the automation output, we introduce a pilot-only baseline state $\bar{x} \in \Reals^{n}$, which evolves according to~\eqref{eq:lindyn} but with no assistance as
\begin{equation}\label{eq:predictor}
	\dot{\bar{x}} = A\,\bar{x} + B_r\,r, \quad \bar{x}(t_0) = x(t_0),
\end{equation}
where $t_0 \ge 0$ is the time at which both the assistance system and the safety filter are engaged (not in \eqref{eq:predictor} but in the closed loop system). Since no automation assistance is applied for $t < t_0$ ($y_a=0$), the state $x(t_0)$ is the pilot-only state at the activation instant, and we initialize the baseline model with $\bar{x}(t_0)=x(t_0)$. Subtracting~\eqref{eq:predictor} from~\eqref{eq:lindyn}, the automation-induced deviation $e \triangleq x - \bar{x}$ satisfies
\begin{equation}\label{eq:error}
	\dot{e} = A\,e + B_{y_a}\,y_a.
\end{equation}

The safety filter in~\eqref{eq:qp} uses two \emph{workload barriers}, $b_1$ and $b_2$, to confine the pilot workload $w_p$ in~\eqref{eq:wp} between time-varying envelope bounds $w_L(t)$ and $w_U(t)$ as
\begin{subequations}\label{eq:b12}
	\begin{align}
		b_1(x,w_L,\rho) &= w_p(x,\rho) - w_L(t), \\
		b_2(x,w_U,\rho) &= w_U(t) - w_p(x,\rho),
	\end{align}
\end{subequations}
so that $b_1 \ge 0$ and $b_2 \ge 0$ enforce $w_L \le w_p \le w_U$. A third barrier, introduced in Section~\ref{sec:devbarrier}, limits the deviation~\eqref{eq:error} between the actual and baseline states. The envelope dynamics governing $w_L$ and $w_U$ are constructed in Section~\ref{sec:envelope}.

\subsection{Deviation Barrier}
\label{sec:devbarrier}

Using~\eqref{eq:pilot_ss_b}, $y_h$ can be written as $y_h = Hx + d_h r$ by defining
\begin{equation}
	H \triangleq [-d_h C_p,\; c_h] \in \Reals^{1 \times n}. \label{eq:H_def}
\end{equation}

\begin{assumption}\label{ass:reldeg}
	The plant~\eqref{eq:plant_lin} satisfies $C_p B_p = 0$ and $C_p A_p B_p \neq 0$, i.e., the output $y_p$ has relative degree two with respect to the input $u_p$.
\end{assumption}

\begin{remark}
	In pitch-angle control where the elevator is the control input, the output has relative degree two.
\end{remark}

Differentiating $y_h = Hx + d_h r$ along~\eqref{eq:lindyn} gives $\dot{y}_h = H(Ax + B_{y_a} y_a + B_r r) + d_h \dot{r}$. If $H B_{y_a} \ne 0$, then $\dot{y}_h$ depends on $y_a$, and so does $w_p$ in~\eqref{eq:wp}. This is undesirable because the safety filter in~\eqref{eq:qp} computes $y_a$ based on $w_p$ (see~\eqref{eq:b12}). If $w_p$ changes with $y_a$, then the safety-filter selects the assistance input $y_a$ using a workload measure that already depends algebraically on that same input. Barriers in~\eqref{eq:b12} therefore no longer provide fixed inequalities for the optimization over $y_a$, but instead depend on the very variable being solved for.

Since $B_{y_a}$ in~\eqref{eq:lindyn} has the form $B_{y_a} = [B_p^\tr,\; 0]^\tr$, using~\eqref{eq:H_def} and Assumption~\ref{ass:reldeg}, it is obtained that
\begin{equation}\label{eq:h0bya}
	H B_{y_a} = [-d_h C_p,\; c_h]\begin{bmatrix} B_p \\ 0 \end{bmatrix} = -d_h C_p B_p = 0.
\end{equation}

Differentiating $\dot{y}_h = H(Ax + B_{y_a} y_a + B_r r) + d_h \dot{r}= HAx + HB_r r + d_h \dot{r}$ and substituting~\eqref{eq:lindyn} with $y_a = 0$ gives $\ddot{y}_h\big|_{y_a=0} = HA(Ax + B_r r) + HB_r \dot{r} + d_h \ddot{r}$. To collect the quantities that appear in the workload definition~\eqref{eq:wp}, its no-assistance evolution under~\eqref{eq:lindyn}, and the subsequent QP construction in~\eqref{eq:qp}, we define three quantities as
\begin{subequations}\label{eq:outputs}
	\begin{align}
		y_0 &\triangleq y_h = h_0^\tr x + d_h\, r, \label{eq:y0} \\
		y_1 &\triangleq \dot{y}_h = h_1^\tr x + H B_r\, r + d_h\, \dot{r}, \label{eq:y1}\\
		y_2 &\triangleq \ddot{y}_h\big|_{y_a=0} = h_2^\tr x + H\!A B_r\, r + H B_r\, \dot{r} + d_h\, \ddot{r}, \label{eq:y2}
	\end{align}
\end{subequations}
where \(h_k \triangleq (H A^k)^\tr \in \Reals^n\) for \(k=0,1,2\), with $A$ in~\eqref{eq:Amat}. We define $S \triangleq [h_0 \;\; h_1 \;\; h_2]^\tr \in \Reals^{3 \times n}$ so that \(S\) maps the state \(x\) to the state-dependent part of these quantities as
\begin{equation}\label{eq:y_stack}
	\begin{bmatrix}
		y_0 \\ y_1 \\ y_2
	\end{bmatrix}
	=
	Sx +
	\begin{bmatrix}
		d_h\, r \\
		H B_r\, r + d_h\, \dot{r} \\
		H A B_r\, r + H B_r\, \dot{r} + d_h\, \ddot{r}
	\end{bmatrix}.
\end{equation}
In this notation, the workload~\eqref{eq:wp} becomes \(w_p = y_0^2 + y_1^2\).

Using~\eqref{eq:outputs} and differentiating $w_p = y_0^2 + y_1^2$ with respect to~$x$ gives $\frac{\partial w_p}{\partial x} = 2\,y_0\,h_0^\tr + 2\,y_1\,h_1^\tr$. Therefore,
\begin{equation}\label{eq:Lgawp_raw}
	L_{g_a} w_p = \frac{\partial w_p}{\partial x} B_{y_a} = 2\,y_0(h_0^\tr B_{y_a}) + 2\,y_1(h_1^\tr B_{y_a}),
\end{equation}
where $B_{y_a}$ is given in~\eqref{eq:lindyn}.

Substituting $h_0^\tr B_{y_a} = H B_{y_a} = 0$ (see~\eqref{eq:h0bya}) into~\eqref{eq:Lgawp_raw}, and differentiating $w_p = y_0^2 + y_1^2$ along the no-assistance dynamics ($y_a = 0$) of~\eqref{eq:lindyn} using~\eqref{eq:outputs}, the Lie derivatives of $w_p$ simplify to
\begin{subequations}\label{eq:wp_derivs}
	\begin{align}
		\mathcal{L}_0 w_p &= 2\,y_0\,y_1 + 2\,y_1\,y_2, \label{eq:L0wp}\\
		L_{g_a} w_p &= 2\,y_1\,(H A B_{y_a}), \label{eq:Lgawp}
	\end{align}
\end{subequations}
where $\mathcal{L}_0 w_p$ is the time derivative of $w_p$ under no assistance ($y_a = 0$) as defined in~\eqref{eq:L0}, $L_{g_a} w_p = \frac{\partial w_p}{\partial x} B_{y_a}$ is the control Lie derivative from~\eqref{eq:La}, $y_0, y_1, y_2$ are defined in~\eqref{eq:outputs}, and $H A B_{y_a} = h_1^\tr B_{y_a}$ with $h_1 = (HA)^\tr$.

We now construct a deviation barrier that limits how far the actual state $x$ in~\eqref{eq:lindyn} can deviate from the baseline state $\bar{x}$ in~\eqref{eq:predictor}. Let $\Wo$ be the solution of the Lyapunov equation
\begin{equation}\label{eq:lyap}
	A^\tr \Wo + \Wo A = -S^\tr S,
\end{equation}
where $A$ is the aggregated system matrix given in~\eqref{eq:Amat} and $S$ is from~\eqref{eq:y_stack}. Since $A$ is Hurwitz, $\Wo \ge 0$. The deviation barrier is defined as
\begin{equation}\label{eq:b3}
	b_3(x,\bar{x}) = \eta^2 - e^\tr \Wo\, e, \quad e = x - \bar{x},
\end{equation}
where $e$ is the deviation defined in~\eqref{eq:error} and $\eta > 0$ is a design parameter. The constraint $b_3 \ge 0$ confines the deviation to the set $C_3 \triangleq \{(x,\bar{x}) : e^\tr \Wo\, e \le \eta^2\}$, an ellipsoidal tube around the baseline trajectory whose size is controlled by $\eta$.

\begin{lemma}\label{lem:b3}
	The function $b_3$ in~\eqref{eq:b3} is an ACBF (see Definition~\ref{def:acbf}) for the coupled dynamics~\eqref{eq:lindyn}, with the class-$\mathcal{K}$ function $\alpha_3(s) = \gamma_3 s$, where $\gamma_3 > 0$ is a design parameter.
\end{lemma}

\begin{proof}
	Under $y_a = 0$, the deviation dynamics~\eqref{eq:error} reduce to $\dot{e} = Ae$. Using~\eqref{eq:lyap} and differentiating $b_3$ from~\eqref{eq:b3} along $\dot{e} = Ae$ gives
	\begin{align}
		\mathcal{L}_0 b_3 &= -2\,e^\tr \Wo\,(Ae)
		= e^\tr(-\Wo A - A^\tr \Wo)\,e \notag\\
		&= e^\tr S^\tr S\, e. \label{eq:L0b3}
	\end{align}
	Furthermore,
	\begin{equation}\label{eq:acbf3_check}
		\mathcal{L}_0 b_3 + \gamma_3\, b_3 = \underbrace{e^\tr S^\tr S\, e}_{\ge\, 0} + \gamma_3\underbrace{(\eta^2 - e^\tr \Wo\, e)}_{\ge\, 0 \text{ on } C_3} \ge 0.
	\end{equation}
	The first term is non-negative because $S^\tr S \ge 0$, and the second is non-negative on $C_3$ by definition of $C_3 = \{(x,\bar{x}) : e^\tr \Wo\, e \le \eta^2\}$. Thus, \eqref{eq:acbf3_check} is of the form~\eqref{eq:acbf} with $\alpha_3(s)=\gamma_3 s$.
\end{proof}

\begin{lemma}\label{lem:range}
	Let $\Wo$ be the solution of~\eqref{eq:lyap}. Then $\mathrm{range}(S^\tr) \subseteq \mathrm{range}(\Wo)$.
\end{lemma}

\begin{proof}
	Since $A$ is Hurwitz, \eqref{eq:lyap} has a unique solution $\Wo$, and because the right-hand side $S^\tr S$ is symmetric, this solution satisfies $\Wo=\Wo^\tr$. Let $z \in \ker(\Wo)$. Using~\eqref{eq:lyap},
	\begin{equation}
		z^\tr S^\tr S z
		= -z^\tr (A^\tr \Wo + \Wo A) z
		= 0,
	\end{equation}
	where the last equality follows from $\Wo z = 0$ and the symmetry $\Wo=\Wo^\tr$. Since
	$z^\tr S^\tr S z = \|Sz\|^2$, it follows that $Sz = 0$. Hence $\ker(\Wo) \subseteq \ker(S)$. Taking orthogonal complements yields
	\begin{equation}
		\mathrm{range}(S^\tr)=\ker(S)^\perp \subseteq \ker(\Wo)^\perp.
	\end{equation}
	Because $\Wo$ is symmetric, $\ker(\Wo)^\perp = \mathrm{range}(\Wo)$. Therefore, $\mathrm{range}(S^\tr) \subseteq \mathrm{range}(\Wo)$.
\end{proof}

\subsection{Mismatch Bound}

We define the mismatch as
\begin{equation}\label{eq:Delta_def}
	\Delta(t) \triangleq \big[\mathcal{L}_0 w_p(x,\rho) + \gamma\, w_p(x,\rho)\big]
	- \big[\mathcal{L}_0 w_p(\bar{x},\rho) + \gamma\, w_p(\bar{x},\rho)\big],
\end{equation}
where $\mathcal{L}_0 w_p$ is the no-assistance drift of the workload from~\eqref{eq:L0wp}, $w_p$ is the workload from~\eqref{eq:wp}, $\bar{x}$ is the baseline state from~\eqref{eq:predictor}, and $\gamma > 0$ is the envelope contraction rate.

\begin{lemma}\label{lem:mismatch}
	Consider the baseline outputs \(y_k=\bar y_k+h_k^\tr e\) for \(k=0,1,2\), obtained by evaluating~\eqref{eq:outputs} at $\bar{x}$ defined in~\eqref{eq:predictor}, where $e = x - \bar{x}$ is defined in~\eqref{eq:error}. Then, the mismatch $\Delta$ in~\eqref{eq:Delta_def} can be rewritten as
	\begin{equation}\label{eq:Delta_expand}
		\Delta = 2\,e^\tr v(\bar{x},\rho) + e^\tr Q_{\mathrm{sym}}\, e,
	\end{equation}
	where
	\begin{subequations}\label{eq:va_defs}
		\begin{align}
			v(\bar{x},\rho) &= \bar{y}_0\, a_0 + \bar{y}_1\, a_1 + \bar{y}_2\, a_2 \in \Reals^n, \label{eq:vt}\\
			a_0 &= h_1 + \gamma\, h_0, \; a_1 = h_0 + h_2 + \gamma\, h_1, \; a_2 = h_1 \in \Reals^n, \label{eq:a_defs}
		\end{align}
	\end{subequations}
	with $h_0, h_1, h_2$ are defined in~\eqref{eq:outputs}. The quadratic coefficient is defined as
	\begin{equation}\label{eq:Qsym}
		Q_{\mathrm{sym}} = \gamma\, h_0 h_0^\tr + h_0 h_1^\tr + h_1 h_0^\tr
		+ \gamma\, h_1 h_1^\tr + h_1 h_2^\tr + h_2 h_1^\tr.
	\end{equation}
\end{lemma}

\begin{proof}
	Substituting $y_k = \bar{y}_k + h_k^\tr e$ into $w_p = y_0^2 + y_1^2$ gives
	$y_k^2 - \bar{y}_k^2 = 2\,\bar{y}_k\,(h_k^\tr e) + (h_k^\tr e)^2$.
	Similarly, for cross terms in $\mathcal{L}_0 w_p = 2 y_0 y_1 + 2 y_1 y_2$,
	$y_i y_j - \bar{y}_i \bar{y}_j = \bar{y}_i(h_j^\tr e) + \bar{y}_j(h_i^\tr e) + (h_i^\tr e)(h_j^\tr e)$.
	Expanding $\Delta$ from~\eqref{eq:Delta_def}, we get $\Delta = \gamma(y_0^2 - \bar{y}_0^2) + \gamma(y_1^2 - \bar{y}_1^2) + 2(y_0 y_1 - \bar{y}_0\bar{y}_1) + 2(y_1 y_2 - \bar{y}_1\bar{y}_2)$. Collecting the linear terms in $e$ yields $\Delta_{\mathrm{lin}} = 2\,e^\tr v(\bar{x},\rho)$ with $v$ as in~\eqref{eq:vt}--\eqref{eq:a_defs}. The quadratic terms give $e^\tr Q_{\mathrm{sym}}\, e$.
\end{proof}

\begin{proposition}\label{prop:bound}
	Under the deviation tube constraint $e^\tr \Wo\, e \le \eta^2$ in~\eqref{eq:b3}, $|\Delta(t)| \le M(\bar{x},\rho) \triangleq 2\,\eta\, L(\bar{x},\rho) + \eta^2\, \rho_Q$, where $\Delta$ is defined in~\eqref{eq:Delta_def}, $\eta$ is the tube radius in~\eqref{eq:b3}, $L(\bar{x},\rho) = \sqrt{\bar{y}^\tr P_g\, \bar{y}}$, $\bar{y} = (\bar{y}_0, \bar{y}_1, \bar{y}_2)^\tr$, $[P_g]_{ij} = a_i^\tr \Wod\, a_j$, $P_g \in \Reals^{3 \times 3}$, $\rho_Q = \rho\!\left(\Wo^{\dagger 1/2}\, Q_{\mathrm{sym}}\, \Wo^{\dagger 1/2}\right)$, where $a_i$ are the coefficient vectors in~\eqref{eq:a_defs}, $Q_{\mathrm{sym}}$ the quadratic matrix in~\eqref{eq:Qsym}, $\Wod$ the Moore--Penrose pseudoinverse of $\Wo$ in~\eqref{eq:lyap}, $\Wo^{\dagger 1/2} \triangleq (\Wod)^{1/2}$ its positive-semidefinite square root, and $\rho(\cdot)$ the spectral radius. Since $\Wo^{\dagger 1/2} Q_{\mathrm{sym}}\, \Wo^{\dagger 1/2}$ is symmetric, $\rho_Q$ equals its induced $2$-norm.
\end{proposition}

\begin{proof}
	Let $\Delta_{\mathrm{lin}} \triangleq 2\,e^\tr v(\bar{x},\rho)$ and $\Delta_{\mathrm{quad}} \triangleq e^\tr Q_{\mathrm{sym}} e$, so that by~\eqref{eq:Delta_expand}, $\Delta=\Delta_{\mathrm{lin}}+\Delta_{\mathrm{quad}}$.

	\textit{Step 1} (Linear term).
	By the Cauchy--Schwarz inequality, we have $|e^\tr v| \le \sqrt{e^\tr \Wo e}\,\sqrt{v^\tr \Wod v}$ provided that $v \in \mathrm{range}(\Wo)$.
	Since $h_0,h_1,h_2 \in \mathrm{range}(S^\tr)$ by definition of $S$ (see~\eqref{eq:y_stack}), Lemma~\ref{lem:range} gives
	$h_0,h_1,h_2 \in \mathrm{range}(\Wo)$. Therefore each coefficient vector
	$a_i$ in~\eqref{eq:a_defs}, being a linear combination of $h_0,h_1,h_2$,
	also belongs to $\mathrm{range}(\Wo)$, and hence
	$v(\bar{x},\rho)=\bar y_0 a_0+\bar y_1 a_1+\bar y_2 a_2 \in \mathrm{range}(\Wo)$.
	Thus $|\Delta_{\mathrm{lin}}|
	= 2|e^\tr v|
	\le 2\sqrt{e^\tr \Wo e}\,\sqrt{v^\tr \Wod v}$. Using the tube constraint $e^\tr \Wo e \le \eta^2$ and
	$v^\tr \Wod v
	= \sum_{i,j=0}^2 \bar y_i \bar y_j\, a_i^\tr \Wod a_j
	= \bar y^\tr P_g \bar y$, we obtain $|\Delta_{\mathrm{lin}}| \le 2\,\eta\,L(\bar{x},\rho)$.

	\textit{Step 2} (Quadratic term).
	By Lemma~\ref{lem:range}, $h_0,h_1,h_2 \in \mathrm{range}(\Wo)$. Since $\Wo$ is symmetric, $\mathrm{range}(\Wo)=\ker(\Wo)^\perp$. Hence, for every $z \in \ker(\Wo)$, $h_i^\tr z = 0$ for $i=0,1,2$. Because $Q_{\mathrm{sym}}$ in~\eqref{eq:Qsym} is a sum of terms of the form $h_i h_j^\tr$, it follows that $Q_{\mathrm{sym}} z = 0$ for all $z \in \ker(\Wo)$. Therefore, $e^\tr Q_{\mathrm{sym}} e$ depends only on the component of $e$ in $\mathrm{range}(\Wo)$. Let $\tilde e = \Wo^{1/2} e$. Then $\|\tilde e\|^2 = e^\tr \Wo e$, and $|\Delta_{\mathrm{quad}}|
	= |e^\tr Q_{\mathrm{sym}} e|
	= |\tilde e^\tr (\Wo^{\dagger 1/2} Q_{\mathrm{sym}} \Wo^{\dagger 1/2}) \tilde e|$. Since $\Wo^{\dagger 1/2} Q_{\mathrm{sym}} \Wo^{\dagger 1/2}$ is symmetric,
	$|\Delta_{\mathrm{quad}}|
	\le \rho_Q \|\tilde e\|^2
	= \rho_Q\, e^\tr \Wo e
	\le \rho_Q\, \eta^2$.

	\textit{Step 3.}
	Combining the two bounds gives
	\begin{equation}\label{eq:proposition1}
		|\Delta|
		\le |\Delta_{\mathrm{lin}}| + |\Delta_{\mathrm{quad}}|
		\le 2\eta L(\bar{x},\rho) + \eta^2 \rho_Q
		= M(\bar{x},\rho). \qedhere
	\end{equation}
\end{proof}

Proposition~\ref{prop:bound} converts the deviation tube constraint $e^\tr \Wo\, e \le \eta^2$ into a scalar workload-rate margin that the envelope dynamics can absorb. The matrices $P_g$ and $\rho_Q$ are precomputed once from the system matrices.

\subsection{Fault-Isolated Envelope Dynamics}
\label{sec:envelope}
We initialize the envelopes $w_L$ and $w_U$ in~\eqref{eq:b12} at the automation engagement time $t_0$. Since the baseline model in~\eqref{eq:predictor} is initialized with $\bar{x}(t_0)=x(t_0)$, the automation-induced deviation in~\eqref{eq:error} at $t_0$ is $e(t_0)=0$. Hence, the workload in~\eqref{eq:wp} is the same with the workload of the baseline model, i.e. $w_p(x(t_0),\rho(t_0)) = w_p(\bar{x}(t_0),\rho(t_0))$. We therefore center the initial envelope at the baseline workload and choose its initial width according to the mismatch bound $M$ (see Proposition~\ref{prop:bound}) as
\begin{equation}
	d(t_0)\triangleq w_U(t_0)-w_L(t_0)=\frac{2M(\bar{x}(t_0),\rho(t_0))}{\gamma},
\end{equation}
and the envelopes are initialized as
\begin{subequations}
	\begin{equation}
		w_L(t_0)=w_p(\bar{x}(t_0),\rho(t_0))-\frac{d(t_0)}{2},
	\end{equation}
	\begin{equation}
		w_U(t_0)=w_p(\bar{x}(t_0),\rho(t_0))+\frac{d(t_0)}{2}.
	\end{equation}
\end{subequations}

The envelope dynamics are
\begin{subequations}\label{eq:envelope}
	\begin{align}
		\dot{w}_L &= \mathcal{L}_0 w_p(\bar{x},\rho) + \gamma\big(w_p(\bar{x},\rho) - w_L\big) - M(\bar{x},\rho), \label{eq:wL}\\
		\dot{w}_U &= \mathcal{L}_0 w_p(\bar{x},\rho) - \gamma\big(w_U - w_p(\bar{x},\rho)\big) + M(\bar{x},\rho), \label{eq:wU}
	\end{align}
\end{subequations}
where $\mathcal{L}_0 w_p(\bar{x},\rho)$ is the no-assistance workload drift from~\eqref{eq:L0wp} evaluated at the baseline state $\bar{x}$ from~\eqref{eq:predictor}, $w_p(\bar{x},\rho)$ is the baseline workload from~\eqref{eq:wp}, $M(\bar{x},\rho)$ is the mismatch bound from Proposition~\ref{prop:bound}, and $\gamma > 0$ is the contraction rate from~\eqref{eq:Delta_def}. Because these dynamics depend on $\bar{x}$ rather than on $x$, the envelope is computable from the baseline model~\eqref{eq:predictor} alone and is decoupled from both the actual state $x$ and the safety filter output $y_a$.

\begin{theorem}\label{thm:workload}
	$b_1$ and $b_2$ that are defined in~\eqref{eq:b12} are valid ACBFs (see Definition~\ref{def:acbf}) for the dynamics~\eqref{eq:lindyn}, with $\alpha_1(s)=\alpha_2(s)=\gamma s$.
\end{theorem}

\begin{proof}
	We must show that $\mathcal{L}_0 b_i(x,\rho)+\gamma\, b_i(x,\rho)\ge 0$ for $i=1,2$ on the corresponding safe sets. For $b_1 = w_p(x,\rho)-w_L$, we have $\mathcal{L}_0 b_1 = \mathcal{L}_0 w_p(x,\rho)-\dot w_L$. Substituting~\eqref{eq:wL}, $\mathcal{L}_0 b_1+\gamma b_1=\mathcal{L}_0 w_p(x,\rho)-\dot w_L+\gamma\bigl(w_p(x,\rho)-w_L\bigr) = \Delta + M(\bar x,\rho) \ge 0$, where $\Delta$ is the mismatch defined in~\eqref{eq:Delta_def}, and the last inequality follows from Proposition~\ref{prop:bound}.

	For $b_2 = w_U-w_p(x,\rho)$, we similarly have $\mathcal{L}_0 b_2 = \dot w_U-\mathcal{L}_0 w_p(x,\rho)$. Substituting~\eqref{eq:wU}, $\mathcal{L}_0 b_2+\gamma b_2=\dot w_U-\mathcal{L}_0 w_p(x,\rho)+\gamma\bigl(w_U-w_p(x,\rho)\bigr) =-\Delta + M(\bar x,\rho) \ge 0$.

	Thus both $b_1$ and $b_2$ satisfy the ACBF condition of Definition~\ref{def:acbf} with
	$\alpha_i(s)=\gamma s$, $i=1,2$.
\end{proof}

\begin{lemma}\label{lem:separation}
	The separation $d(t) \triangleq w_U(t) - w_L(t)$, where $w_U$ and $w_L$ are the envelope bounds from~\eqref{eq:envelope}, satisfies $d(t) \ge 0$ for all $t \ge t_0$, and is bounded.
\end{lemma}

\begin{proof}
	Subtracting~\eqref{eq:wL} from~\eqref{eq:wU} gives $\dot{d} = -\gamma\, d + 2\,M(\bar{x},\rho)$, where $\gamma$ and $M$ are introduced in~\eqref{eq:Delta_def} and Proposition~\ref{prop:bound}. The explicit solution is given as
	\begin{equation}\label{eq:d_sol}
		d(t) = e^{-\gamma(t - t_0)} d(t_0) + 2\!\int_{t_0}^{t} e^{-\gamma(t-\tau)} M(\bar{x}(\tau),\rho(\tau))\, d\tau.
	\end{equation}
	Since $d(t_0) \ge 0$ and $M(\bar{x},\rho) \ge 0$ (because $M = 2\eta L + \eta^2 \rho_Q$ with $\eta > 0$, $L \ge 0$, and $\rho_Q \ge 0$), all terms in~\eqref{eq:d_sol} are non-negative. Consequently, $w_U(t) \ge w_L(t)$ is maintained for all time. For boundedness, since $A$ is Hurwitz and $r$ is bounded, $\bar{x}$ is bounded, hence $L(\bar{x},\rho)$ is bounded, and therefore $M(\bar{x},\rho) \le M_{\sup} \triangleq \sup_{t \ge t_0} M(\bar{x},\rho) < \infty$. Then, $d(t) \le e^{-\gamma(t-t_0)} d(t_0) + 2M_{\sup}/\gamma$, which is bounded.
\end{proof}

\subsection{Quadratic Program (QP) with Three Barriers}

The three barriers, $b_1$, $b_2$ from~\eqref{eq:b12} (workload bounds) and $b_3$ from~\eqref{eq:b3} (deviation tube) jointly depend on the augmented state $z \triangleq (x,\, \bar{x},\, w_L,\, w_U) \in \Reals^{2n+2}$, where $x$ is the augmented state from~\eqref{eq:lindyn}, $\bar{x}$ is the baseline state from~\eqref{eq:predictor}, and $w_L, w_U$ are the envelope bounds from~\eqref{eq:envelope}. The augmented dynamics collect~\eqref{eq:lindyn}, \eqref{eq:predictor}, and~\eqref{eq:envelope} as
\begin{equation}\label{eq:augdyn}
	\dot{z} = F_0(z,\rho) + G(z)\, y_a,
\end{equation}
where $F_0$ stacks $Ax + B_r r$ from~\eqref{eq:lindyn}, $A\bar{x} + B_r r$ from~\eqref{eq:predictor}, $\dot{w}_L$ from~\eqref{eq:wL}, and $\dot{w}_U$ from~\eqref{eq:wU} (all evaluated at $y_a = 0$), and $G = \begin{bmatrix} B_{y_a}^\tr & 0_{1 \times n} & 0 & 0 \end{bmatrix}^\tr \in \Reals^{(2n+2)\times 1}$ with $B_{y_a}$ from~\eqref{eq:lindyn}. The safe set is
\begin{equation}\label{eq:safeset}
	\mathcal{C}(\rho) = \{z : b_1(x,w_L,\rho) \ge 0\} \cap \{z : b_2(x,w_U,\rho) \ge 0\} \cap C_3,
\end{equation}
where $b_1, b_2$ are defined in~\eqref{eq:b12} and $C_3$ is given after~\eqref{eq:b3}.

The safety filter solves the QP~\eqref{eq:qp} using~\eqref{eq:b12} and~\eqref{eq:b3} as
\begin{align}\label{eq:qp_full}
	y_a^\star &= \arg\min_{y_a \in Y_a}\; \|y_a - y_{a,\mathrm{nom}}\|^2 \\
	\text{s.t.}\;\; & \mathcal{L}_a b_1 + \gamma\, b_1 \ge 0, \quad \mathcal{L}_a b_2 + \gamma\, b_2 \ge 0, \quad \mathcal{L}_a b_3 + \gamma_3\, b_3 \ge 0, \notag
\end{align}
where $Y_a \subseteq \mathbb{R}$ is the admissible assistance set, assumed compact and satisfying $0 \in Y_a$, $y_{a,\mathrm{nom}}$ is the unfiltered automation output, $\gamma > 0$ and $\gamma_3 > 0$ are introduced in~\eqref{eq:envelope} and Lemma~\ref{lem:b3}, respectively, and $\mathcal{L}_a$ is given in~\eqref{eq:La}. Each constraint is affine in $y_a$ (see~\eqref{eq:Lgawp}, \eqref{eq:b3} and \eqref{eq:error}), which makes~\eqref{eq:qp_full} a scalar QP.

\begin{theorem}\label{thm:invariance}
	Let $y_a^\star$ be defined by the QP~\eqref{eq:qp_full}. If $z(t_0)$ in~\eqref{eq:augdyn} is in the safe set $\mathcal{C}(\rho(t_0))$ in~\eqref{eq:safeset}, then~\eqref{eq:augdyn} with $y_a = y_a^\star$ admits a solution for all $t \ge t_0$, and the safe set $\mathcal{C}(\rho(t))$ is forward invariant, that is, $z(t) \in \mathcal{C}(\rho(t))$, $\forall t \ge t_0$.
\end{theorem}

\begin{proof}
	We must show that the QP~\eqref{eq:qp_full} remains feasible, the augmented closed-loop solution exists, and each barrier remains nonnegative for all $t>t_0$.

	First, by Lemma~\ref{lem:b3} and Theorem~\ref{thm:workload}, the three barriers $b_1$, $b_2$, and $b_3$ satisfy the ACBF condition at $y_a = 0$. Since $0 \in Y_a$, the input $y_a = 0$ satisfies all three constraints in~\eqref{eq:qp_full}. Therefore, the QP is feasible $\forall t \ge t_0$.
	Because $Y_a$ is compact and $y_a^\star(t) \in Y_a$ by construction, there exists a constant $\bar{y}_a > 0$ such that $|y_a^\star(t)| \le \bar{y}_a$, $\forall t \ge t_0$. Since $A$ is Hurwitz (see~\eqref{eq:Amat}) and $r$ is bounded, $\bar{x}(t)$ in~\eqref{eq:predictor} is bounded for all $t \ge t_0$. Hence, the baseline outputs $\bar y_0$, $\bar y_1$, $\bar y_2$ are bounded, and therefore the mismatch bound $M(\bar{x},\rho)$ from Proposition~\ref{prop:bound} is bounded as well. The state satisfies ${\dot{x} = Ax + B_{y_a} y_a^\star + B_r r}$. Since $A$ is constant, $r$ is bounded, and $y_a^\star$ is bounded, the forcing term $B_{y_a} y_a^\star + B_r r$ is bounded. Therefore, $x(t)$ exists for all $t \ge t_0$. The envelope states satisfy the linear scalar dynamics~\eqref{eq:envelope}, whose forcing terms are bounded because $\bar{x}$, $\rho$, and $M(\bar{x},\rho)$ are bounded. Hence $w_L(t)$ and $w_U(t)$ also exist for all $t \ge t_0$. Consequently, the augmented state $z=(x,\bar{x},w_L,w_U)$ exists for all $t \ge t_0$.

	Finally, along the closed-loop trajectory generated by~\eqref{eq:qp_full}, the QP constraints enforce ${\dot b_1 + \gamma b_1 \ge 0}$, ${\dot b_2 + \gamma b_2 \ge 0}$, and ${\dot b_3 + \gamma_3 b_3 \ge 0}$, that is, $\dot b_1 \ge -\gamma b_1$, $\dot b_2 \ge -\gamma b_2$, and $\dot b_3 \ge -\gamma_3 b_3$. Since $z(t_0) \in \mathcal{C}(\rho(t_0))$, we have $b_1(t_0)\ge 0$, $b_2(t_0)\ge 0$, and $b_3(t_0)\ge 0$. Therefore, none of the barriers can cross zero from above, and thus $b_1(t)\ge 0$, $b_2(t)\ge 0$, and $b_3(t)\ge 0$ for all $t \ge t_0$. Therefore $z(t) \in \mathcal{C}(\rho(t))$ for all $t \ge t_0$, proving forward invariance.
\end{proof}

\section{Case Study: F-16 Pitch Tracking}
\label{sec:casestudy}

\subsection{Shared Control System Description}

Consider a pitch-tracking task. The linearized longitudinal dynamics of an F-16~\cite{farre2007} about a level-flight trim condition are given by~\eqref{eq:plant_lin} with
\begin{subequations}\label{eq:f16}
	\begin{align}
		A_p^{\mathrm{ol}} &= \begin{bmatrix}
			-0.166 & -10.71 & -7.282 & -32.17 \\
			-0.002 & -0.098 & 0.928 & 0 \\
			0 & -0.625 & -0.467 & 0 \\
			0 & 0 & 1 & 0
		\end{bmatrix}\!, \\
		B_p &= \begin{bmatrix} 4.048 \\ 0.025 \\ 0.899 \\ 0 \end{bmatrix}\!, \quad
		C_p = \begin{bmatrix} 0 & 0 & 0 & 1 \end{bmatrix}\!,
	\end{align}
\end{subequations}
where $A_p^{\mathrm{ol}}$ denotes the open-loop plant matrix and $x_p = [\Delta V_T,\allowbreak\; \Delta\alpha,\allowbreak\; \Delta q,\allowbreak\; \Delta\theta]^\tr$ with units $[\mathrm{ft/s},\allowbreak\;\mathrm{rad},\allowbreak\;\mathrm{rad/s},\allowbreak\;\mathrm{rad}]$. Here, $\Delta$ denotes perturbation from the trim condition. The output $y_p = C_p x_p = \Delta\theta$ is the pitch angle, and the input $u_p$ is the elevator deflection (rad, sign-inverted so that positive $u_p$ produces nose-down moment). We apply state feedback $A_p = A_p^{\mathrm{ol}} - B_p K_p$ with a linear quadratic regulator (LQR) gain $K_p \in \Reals^{1 \times 4}$ computed from $Q_{\mathrm{lqr}} = \mathrm{diag}(4 \times 10^{-4},\,100,\,4,\,25)$ and $R_{\mathrm{lqr}} = 1/0.35^2$, penalising angle-of-attack deviations heavily while limiting elevator authority to approximately $\pm 0.35$~rad.

The human pilot is described by~\eqref{eq:pilot_ss}, with parameters $k_p = 3$, $T_p = 0.2$~s, $T_z = 0.6$~s. The aggregated system~\eqref{eq:lindyn} has dimension $n = 5$.

\subsection{Simulations}

The assistance system is selected as a replica of the human pilot model, aiming to enhance pilot performance and reduce workload. The reference signal is $r(t) = 0.1(\cos 0.5t - \sin t)$~rad. No automation assistance is applied for $t < t_0 = 10$~s. At $t_0$, both the assistance system described by the replica of the human model and the safety filter in~\eqref{eq:qp_full} are engaged simultaneously. The design parameters are $\gamma = 5$ (see~\eqref{eq:envelope}), $\gamma_3 = 10$ (see Lemma~\ref{lem:b3}), and $\eta = 0.08$ (see~\eqref{eq:b3}). Apart from the nominal assistance, we examine four fault scenarios representing distinct failure modes.

\textbf{Case~1 (Drifting assistance):} For $t \ge t_0$, the nominal assistance command is replaced by
$y_{a,\mathrm{nom}}(t) = -0.4 - 0.02(t - t_0)$, i.e., a linearly growing bias.

\textbf{Case~2 (Frozen assistance):} For $t \ge t_0$, the nominal assistance command is frozen at
$y_{a,\mathrm{nom}}(t) = 0.6$~rad.

\textbf{Case~3 (Sign-reversed assistance):} For $t \ge t_0$, the nominal assistance command is multiplied by $-1$, so the automation opposes the pilot.

\textbf{Case~4 (Aggressive assistance):} For $t \ge t_0$, the nominal assistance command is amplified by a factor of $35$, i.e.\ $y_{a,\mathrm{nom}}(t) \mapsto 35\,y_{a,\mathrm{nom}}(t)$. Because the nominal assistance is a pilot-output replica $y_{a,\mathrm{nom}} = H x + d_h r$ (with $H$ from~\eqref{eq:H_def}), the closed-loop matrix becomes $A + 35\, B_{y_a} H$, which is unstable at this gain.

\begin{figure}[tbp]
	\centering
	\includegraphics[width=0.62\textwidth]{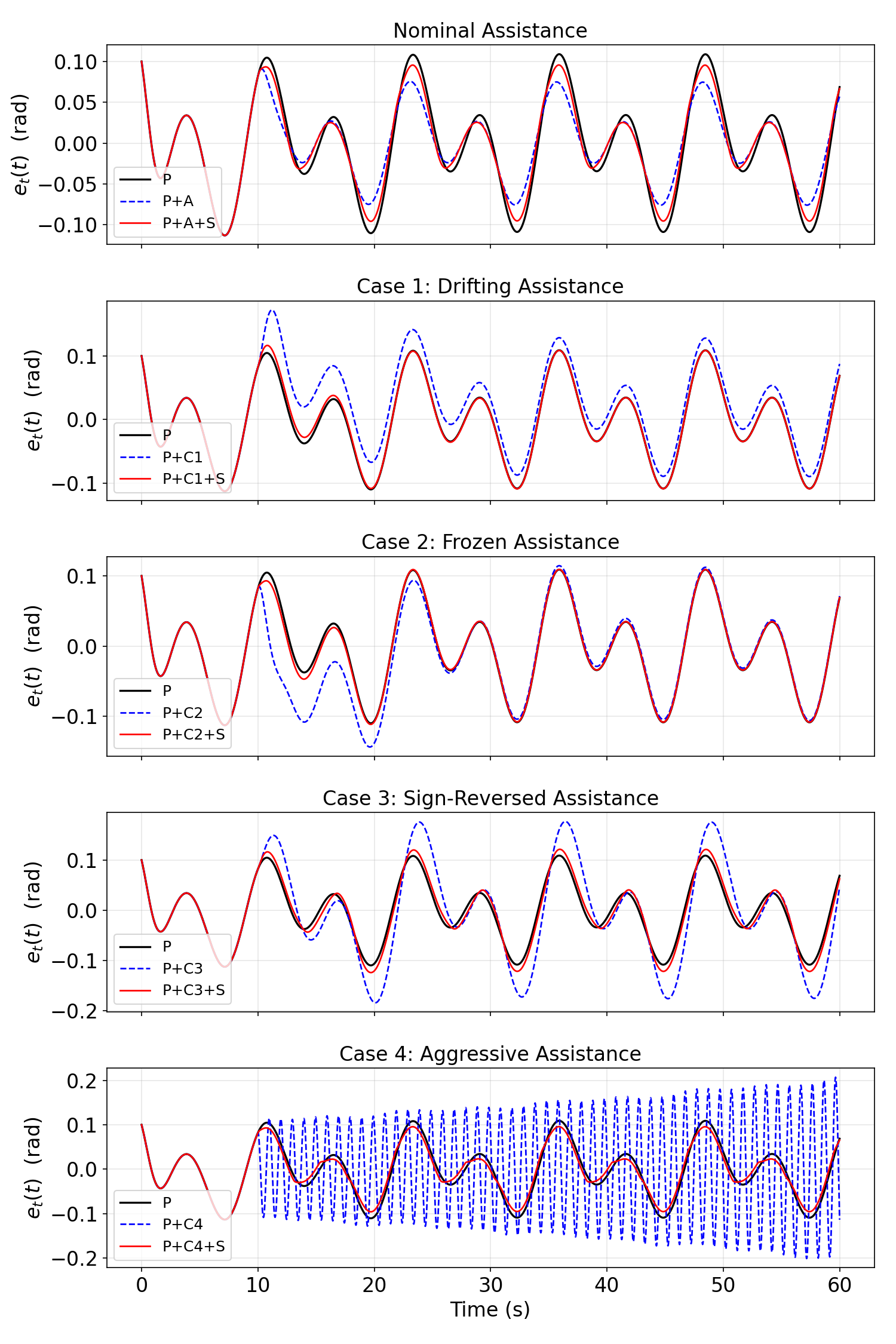}
	\caption{Tracking error $e_t(t) = r(t) - y_p(t)$ under nominal conditions and four fault scenarios (Cases~1--4). Each subplot shows the pilot-only baseline (P, black solid), unfiltered assistance (P+A or P+Ci, blue dashed), and safety-filtered assistance (P+A+S or P+Ci+S, red solid). P = pilot only, A = nominal assistance, S = safety filter, Ci = fault Case~$i$.}
	\label{fig:et}
\end{figure}

Fig.~\ref{fig:et} shows the tracking error $e_t \triangleq r - y_p$ across all scenarios. Nominal assistance improves tracking performance by reducing the tracking error, while the proposed method limits this improvement to enforce the barrier constraints in~\eqref{eq:qp_full}. Across all four fault scenarios, unfiltered assistance degrades tracking to varying degrees, but safety enforcement confines the degradation near the pilot-only baseline.

\begin{figure}[tbp]
	\centering
	\includegraphics[width=0.62\textwidth]{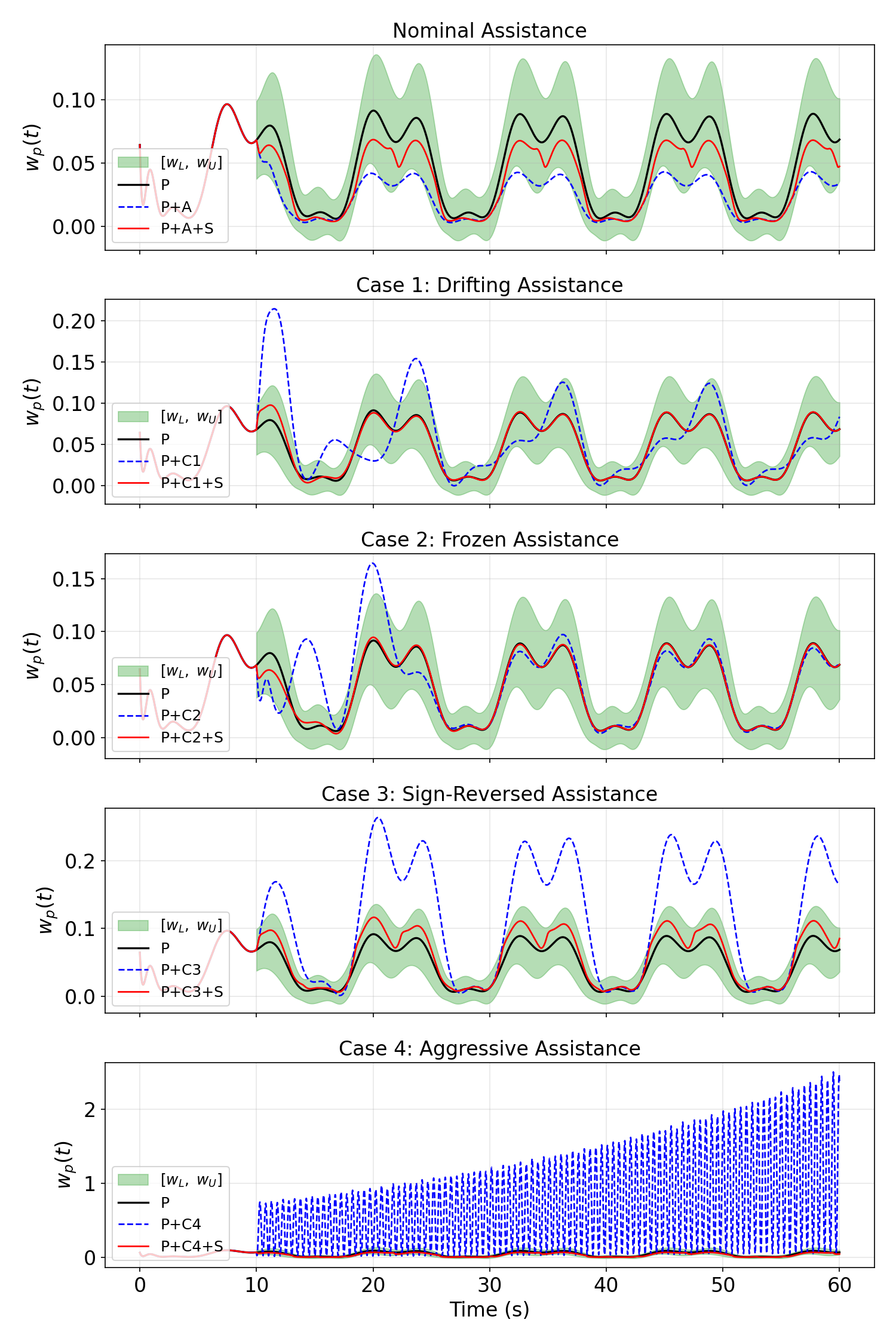}
	\caption{Workload $w_p$ defined in~\eqref{eq:wp} with pilot-anchored adaptive envelope $[w_L(t),\, w_U(t)]$ (see~\eqref{eq:envelope}) under nominal conditions and four fault scenarios (Cases~1--4). Line styles: P (black solid), P+A or P+Ci (blue dashed), P+A+S or P+Ci+S (red solid), envelope (green shaded). Notation: P = pilot only, A = nominal assistance, S = safety filter, Ci = fault Case~$i$.}
	\label{fig:env}
\end{figure}

Fig.~\ref{fig:env} shows the workload profiles with the pilot-anchored envelope. Without safety enforcement, nominal assistance diminishes workload significantly, risking loss of situational awareness~\cite{endsley1995ootl}. Each fault causes workload to deviate from the pilot-only baseline, most severely in Case~4 (destabilized closed loop). Safety enforcement keeps workload near the pilot-only case across all scenarios. The pilot-anchored bounds $w_L$ and $w_U$ governed by~\eqref{eq:envelope} depend only on $\bar{x}$ and $r$. These bounds (envelopes) adapt to the task phase while remaining decoupled from the safety filter output $y_a^\star$ in~\eqref{eq:qp_full}. The envelope shape is identical across all five scenarios (nominal and four fault scenarios) because the baseline trajectory $\bar{x}$ is generated by the same pilot-only dynamics and the same reference signal in every case, independent of the automation fault.

\begin{figure}[tbp]
	\centering
	\includegraphics[width=0.62\textwidth]{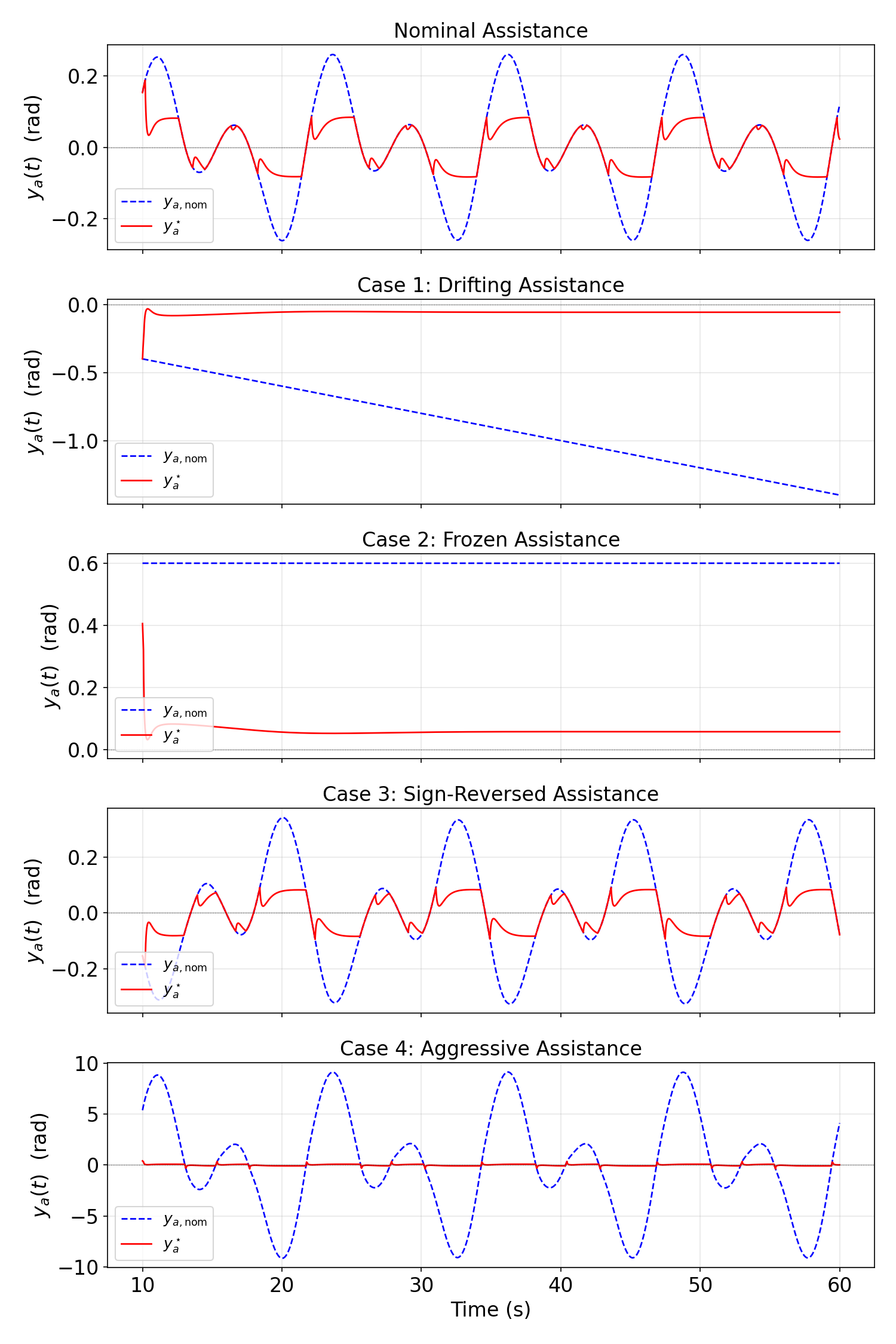}
	\caption{Nominal assistance $y_{a,\mathrm{nom}}(t)$ (blue dashed) and filtered assistance $y_a^\star(t)$ (red solid) (see~\eqref{eq:qp_full}) under nominal conditions and four fault scenarios (Cases~1--4). Both signals are from the same (filtered) simulation run, showing the actual QP input--output pair (see~\eqref{eq:qp_full}). Each row corresponds to one scenario (nominal, then Cases~1--4).}
	\label{fig:ya}
\end{figure}

Fig.~\ref{fig:ya} shows the filtered automation assistance across all scenarios. The workload barriers $b_1$ and $b_2$ in~\eqref{eq:b12} confine the workload between the envelope bounds $w_L$ and $w_U$, while the deviation barrier $b_3$ in~\eqref{eq:b3} is the active constraint most of the time, as the tight tube $\eta = 0.08$ makes confining the automation-induced deviation the primary demand on the safety filter. Under Cases~1 and~2, the safety filter in~\eqref{eq:qp_full} attenuates the large fault signals toward zero, while under Case~3 it moderates the sign-reversed assistance and under Case~4 it attenuates the effect of the destabilizing gain. In all cases, $b_3$ in~\eqref{eq:b3} drives $y_a^\star$ in~\eqref{eq:qp_full} toward zero, thereby maintaining $e^\tr \Wo\, e \le \eta^2$ (the tube constraint from~\eqref{eq:b3} on which the workload barriers $b_1$ and $b_2$ in~\eqref{eq:b12} depend).

\subsection{Quantitative Metrics}

Table~\ref{tab:metrics} summarises quantitative performance across all scenarios (nominal case and four fault scenarios) for $t \ge t_0$. The columns report the percentage change in root mean square (RMS) tracking error relative to the pilot-only baseline, RMS and peak workload deviation from the pilot-only baseline, and the percentage of time the safety filter is active.

\begin{table}[tb]
	\centering
	\caption{Performance metrics for all scenarios ($t \ge t_0 = 10$~s).}
	\label{tab:metrics}
	\begin{tabular*}{0.9\textwidth}{@{\extracolsep{\fill}}l c c c c@{}}
		\toprule
		Scenario & $\Delta\mathrm{RMS}(e_t)$ & $\mathrm{RMS}(\Delta w_p)$ & $\mathrm{Peak}(\Delta w_p)$ & Filter active \\
		& (\%)                       &                              &                               & (\%) \\
		\midrule
		P                    & 0.0613\,\text{rad} & ---    & ---    & --- \\
		P+A                  & $-$30.2 & 0.0294 & 0.0501 & --- \\
		P+A+S                & $-$12.9 & 0.0131 & 0.0230 & 57.7 \\
		\midrule
		P+C1                 & $+$13.8 & 0.0377 & 0.1394 & --- \\
		P+C1+S               & $+$0.9  & 0.0040 & 0.0187 & 100.0 \\
		\midrule
		P+C2                 & $+$7.5  & 0.0244 & 0.0849 & --- \\
		P+C2+S               & $-$0.5  & 0.0035 & 0.0162 & 100.0 \\
		\midrule
		P+C3                 & $+$61.9 & 0.0919 & 0.1756 & --- \\
		P+C3+S               & $+$11.1 & 0.0137 & 0.0250 & 67.1 \\
		\midrule
		P+C4                 & $+$76.7 & 0.9019 & 2.4339 & --- \\
		P+C4+S               & $-$13.3 & 0.0136 & 0.0232 & 97.2 \\
		\bottomrule
	\end{tabular*}

	\vspace{4pt}
	\parbox{0.9\textwidth}{\footnotesize P = pilot only, A = nominal assistance, S = safety filter, C1--C4 = fault cases. $\Delta\mathrm{RMS}(e_t)$: percent change relative to P baseline, where the P row reports the absolute baseline RMS in rad. $\Delta w_p \triangleq w_p(x,\rho) - w_p(\bar{x},\rho)$. Filter active = percentage of time $y_a^\star \ne y_{a,\mathrm{nom}}$, computed for $t \ge t_0 = 10$~s.}
\end{table}

Table~\ref{tab:metrics} shows that safety enforcement consistently reduces both RMS and peak workload deviation from the pilot-only baseline. The most dramatic result is Case~4 (aggressive assistance): the unfiltered $35\times$ gain destabilizes the closed loop and causes tracking error to grow by $77\%$, yet with safety enforcement the tracking error is reduced by $13\%$ relative to the pilot-only baseline. Without safety enforcement, the faults degrade tracking error by $8$--$77\%$ relative to the pilot-only baseline, yet with safety enforcement the tracking error stays within $14\%$ of the baseline across all cases. The safety filter is active $67$--$100\%$ of the time across fault scenarios.

\subsection{Fault Scenario Comparison}

\begin{figure}[tbp]
	\centering
	\includegraphics[width=0.68\textwidth]{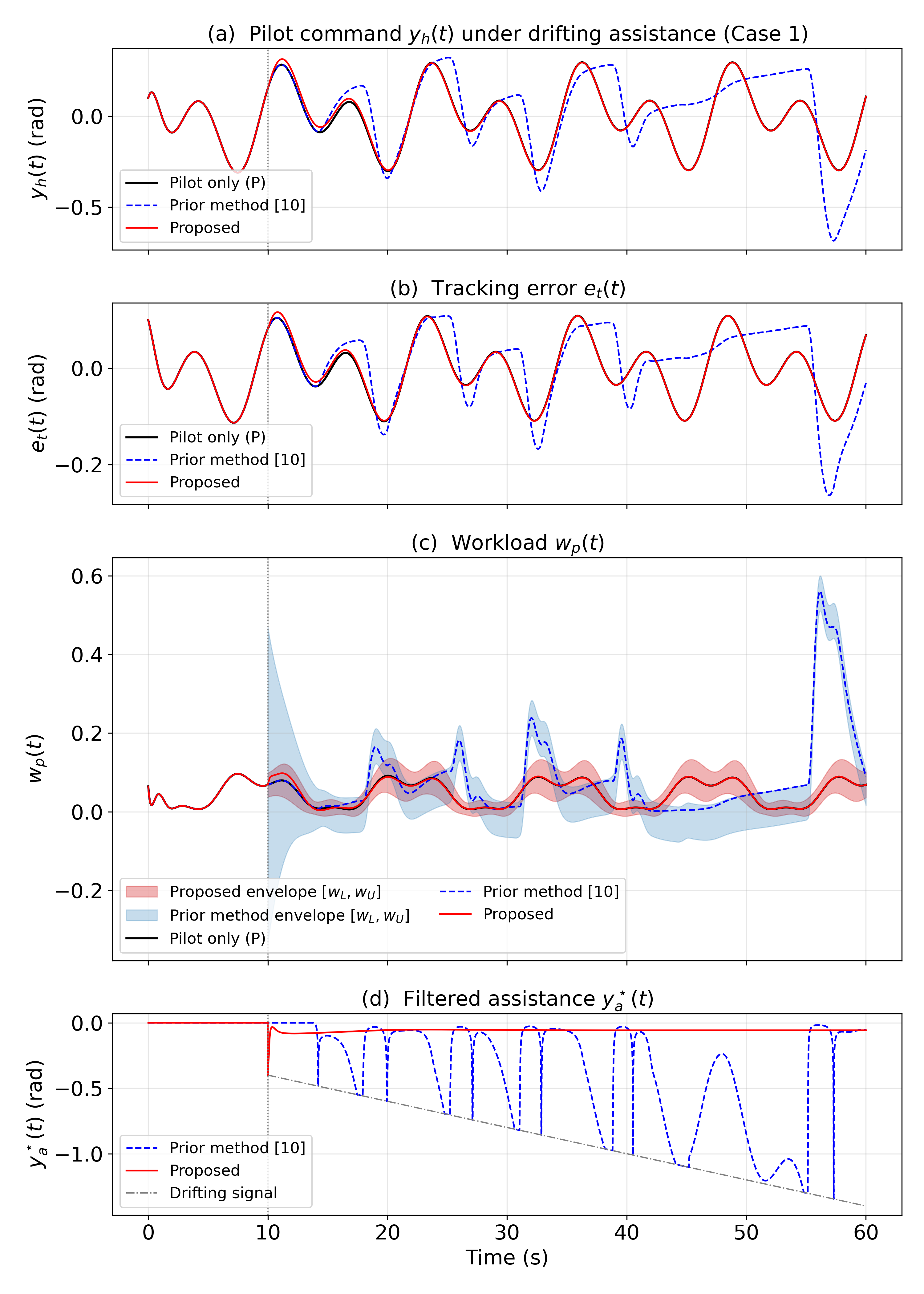}
	\caption{Comparison of the prior method~\cite{uzun2025arbitration} (blue dashed) and proposed pilot-anchored method (red solid) under drifting assistance (Case~1: $y_{a,\mathrm{nom}}(t) = -0.4 - 0.02(t - t_0)$~rad). Pilot-only baseline shown in black solid. (a)~Pilot command $y_h(t)$. (b)~Tracking error $e_t(t)$. (c)~Workload $w_p(t)$ with both envelopes. (d)~Filtered assistance $y_a^\star(t)$, gray dash-dot shows $y_{a,\mathrm{nom}}(t)$, the drifting fault signal.}
	\label{fig:comparison}
\end{figure}

Fig.~\ref{fig:comparison} compares the prior method~\cite{uzun2025arbitration} and the proposed method under drifting assistance (Case~1). The quadratic program (QP) constraint of~\cite{uzun2025arbitration} accommodates the growing command because its envelope dynamics are evaluated on the actual state, so the workload bounds shift progressively with the corrupted trajectory. The state-evaluated envelope evolves with the corrupted state, while the proposed pilot-anchored envelope, evaluated at the baseline state $\bar{x}$ in~\eqref{eq:predictor}, is unaffected. The deviation barrier $b_3$ given in~\eqref{eq:b3} drives the filtered assistance toward zero, maintaining ${|\Delta(t)| \le M(\bar{x},\rho)}$ (see~\eqref{eq:proposition1}).

\section{Conclusion}
\label{sec:conclusion}

This paper identified a structural vulnerability in shared-control arbitration. When the quantities that regulate automation authority are evaluated on the assisted trajectory (when the automation is active), unsafe assistance can corrupt the trajectory and drag the constraint boundaries with it. The proposed framework eliminates this coupling by anchoring the safety envelope to an assistance-independent baseline and confining the actual trajectory within it through a deviation barrier. Simulations on an F-16 pitch-tracking task under four severe unsafe assistance scenarios demonstrate that the proposed method prevents the workload envelope from being dragged by unsafe assistance.

\bibliographystyle{unsrt}
\bibliography{references}

\end{document}